\documentclass[conference,letterpaper,romanappendices]{ieeeconf}
\let\proof\relax   

\usepackage{amsthm,xpatch}
\usepackage{amsmath}
\usepackage{cite}
\usepackage{amssymb}
\usepackage{amsfonts}
\usepackage{dsfont}
\usepackage{graphicx, subfigure}
\usepackage{color}
\usepackage{breqn}
\usepackage{mathtools}
\usepackage{bbm}
\usepackage{latexsym}
\usepackage[ruled, linesnumbered]{algorithm2e}
\usepackage{accents}

\newtheorem{lemma}{Lemma}
\newtheorem{theorem}{Theorem}

\newtheorem{definition}{Definition}

\IEEEoverridecommandlockouts

\begin{document}

\newcommand{\SB}[3]{
\sum_{#2 \in #1}\biggl|\overline{X}_{#2}\biggr| #3
\biggl|\bigcap_{#2 \notin #1}\overline{X}_{#2}\biggr|
}

\newcommand{\Mod}[1]{\ (\textup{mod}\ #1)}

\newcommand{\overbar}[1]{\mkern 0mu\overline{\mkern-0mu#1\mkern-8.5mu}\mkern 6mu}

\makeatletter
\newcommand*\nss[3]{%
  \begingroup
  \setbox0\hbox{$\m@th\scriptstyle\cramped{#2}$}%
  \setbox2\hbox{$\m@th\scriptstyle#3$}%
  \dimen@=\fontdimen8\textfont3
  \multiply\dimen@ by 4             
  \advance \dimen@ by \ht0
  \advance \dimen@ by -\fontdimen17\textfont2
  \@tempdima=\fontdimen5\textfont2  
  \multiply\@tempdima by 4
  \divide  \@tempdima by 5          
  \ifdim\dimen@<\@tempdima
    \ht0=0pt                        
    \@tempdima=\fontdimen5\textfont2
    \divide\@tempdima by 4          
    \advance \dimen@ by -\@tempdima 
    \ifdim\dimen@>0pt
      \@tempdima=\dp2
      \advance\@tempdima by \dimen@
      \dp2=\@tempdima
    \fi
  \fi
  #1_{\box0}^{\box2}%
  \endgroup
  }
\makeatother

\makeatletter
\renewenvironment{proof}[1][\proofname]{\par
  \pushQED{\qed}%
  \normalfont \topsep6\p@\@plus6\p@\relax
  \trivlist
  \item[\hskip\labelsep
        \itshape
    #1\@addpunct{:}]\ignorespaces
}{%
  \popQED\endtrivlist\@endpefalse
}
\makeatother

\makeatletter
\newsavebox\myboxA
\newsavebox\myboxB
\newlength\mylenA

\newcommand*\xoverline[2][0.75]{%
    \sbox{\myboxA}{$\m@th#2$}%
    \setbox\myboxB\null
    \ht\myboxB=\ht\myboxA%
    \dp\myboxB=\dp\myboxA%
    \wd\myboxB=#1\wd\myboxA
    \sbox\myboxB{$\m@th\overline{\copy\myboxB}$}
    \setlength\mylenA{\the\wd\myboxA}
    \addtolength\mylenA{-\the\wd\myboxB}%
    \ifdim\wd\myboxB<\wd\myboxA%
       \rlap{\hskip 0.5\mylenA\usebox\myboxB}{\usebox\myboxA}%
    \else
        \hskip -0.5\mylenA\rlap{\usebox\myboxA}{\hskip 0.5\mylenA\usebox\myboxB}%
    \fi}
\makeatother

\xpatchcmd{\proof}{\hskip\labelsep}{\hskip3.75\labelsep}{}{}

\pagestyle{plain}

\title{\fontsize{22.59}{28}\selectfont A Monetary Mechanism for Stabilizing Cooperative Data Exchange with Selfish Users}

\author{Anoosheh Heidarzadeh, Ishan Tyagi, Srinivas Shakkottai, and Alex Sprintson\thanks{The authors are with the Department of Electrical and Computer Engineering, Texas A\&M University, College Station, TX 77843 USA (E-mail: \{anoosheh, ishantyagi1992, sshakkot, spalex\}@tamu.edu).}}


\maketitle 

\thispagestyle{plain}

\begin{abstract}
This paper considers the problem of stabilizing cooperative data exchange with selfish users. In this setting, each user has a subset of packets in the ground set $X$, and wants all other packets in $X$. The users can exchange their packets by broadcasting coded or uncoded packets over a lossless broadcast channel, and monetary transactions are allowed between any pair of users. We define the utility of each user as the sum of two sub-utility functions: (i) the difference between the total payment received by the user and the total transmission rate of the user, and (ii) the difference between the total number of required packets by the user and the total payment made by the user. A rate-vector and payment-matrix pair $(r,p)$ is said to stabilize the grand coalition (i.e., the set of all users) if $(r,p)$ is Pareto optimal over all minor coalitions (i.e., all proper subsets of users who collectively know all packets in $X$). Our goal is to design a stabilizing rate-payment pair with minimum total sum-rate and minimum total sum-payment for any given instance of the problem. In this work, we propose two algorithms that find such a solution. Moreover, we show that both algorithms maximize the sum of utility of all users (over all solutions), and one of the algorithms also maximizes the minimum utility among all users (over all solutions).   
\end{abstract}

\section{introduction}
Over the last decade, several variants of the cooperative data exchange (CDE) problem have been studied in the literature, see, e.g.,~\cite{RCS:2007,RSS:2010,CW:2010,SSBR:2010,SSBR2:2010,GL:2012,YSZ:2014,CW:2014,CH:2014, MPRGR:2016, CAZDLS:2016}. The original setting of this problem considers a peer-to-peer data exchange scenario over a lossless broadcast channel. There is a group $N$ of users and a ground set $X$ of packets. Each user knows a subset of packets in $X$, and wants to learn the rest of packets in $X$. The users exchange their packets by broadcasting coded or uncoded versions of their packets, and the problem is to find a solution (i.e., the transmission rate of each user and the set of packets transmitted by each user) such that all users achieve omniscience with minimum total sum-rate. 

In this work, we revisit the CDE problem from a game-theoretic perspective where all users are selfish. In this setting, there can be a monetary transaction between any pair of users, and the utility function of each user is defined as the sum of two sub-utility functions as follows: (i) the difference between the total payment the user receives from other users and its transmission rate, and (ii) the difference between the total number of packets the user wants and the total payment it makes to other users. Thinking of the sum of the transmission rate and the total payment being made by each user as its \emph{cost} for participating in the exchange session, and thinking of the sum of the number of packets each user learns and the total payment being received by the user as its \emph{gain} due to its participation in the exchange session, the utility function of each user is the surplus of the user. 

The problem is to find a rate schedule $\{r_i\}_{i\in N}$ and a payment schedule $\{p_{i,j}\}_{i,j\in N}$ for the \emph{grand coalition} (i.e., the set of all users) to achieve omniscience all together that is Pareto optimal, with respect to the utility function, over all \emph{minor coalitions} (i.e., any proper subset of users who collectively know all packets in $X$). That is, a pair $(\{r_i\}_{i\in N},\{p_{i,j}\}_{i,j\in N})$ is a \emph{solution} if there is no pair $(\{\tilde{r}_i\}_{i\in S},\{\tilde{p}_{i,j}\}_{i,j\in S})$ for any minor coalition $S$ to achieve omniscience together such that the utility of some user(s) in $S$ is strictly greater, and the utility of no user in $S$ is less. Note that a solution \emph{stabilizes} the grand coalition in that no minor coalition has incentive to break the grand coalition. The goal is to find a solution that minimizes the total sum-rate and the total sum-payment simultaneously. 

In this work, we propose two algorithms, each of which finds a solution for any problem instance. Moreover, we show that both algorithms maximize the sum of utility of all users (over all solutions), and one of the algorithms also maximizes the minimum utility among all users (over all solutions).  

\subsection{Related Work}
A different coalition-game model for the CDE problem was recently proposed in~\cite{DCLKS:16}. This model, however, differs from our work in two aspects: (i) the utility function under the consideration is different from ours, and (ii) the criteria for the stability of the grand coalition is different from the Pareto optimality being considered here. 

Very recently, in~\cite{HS3:2016}, we also studied a related problem, where each user has two utility functions: its rate and its delay. Defining the stability of the grand coalition via the Pareto optimality, with respect to both the rate and delay functions simultaneously, over all minor coalitions, we showed that there does not exist any \emph{non-monetary} mechanism (without the peer-to-peer payments) that stabilizes the grand coalition for all problem instances. This result is the motivation of this work on the design of a \emph{monetary} mechanism for stabilizing the grand coalition for any problem instance.

\section{Problem Setup}
We consider the original setting of the cooperative data exchange (CDE) problem as follows. Consider a group of $n$ users and a set of $k$ packets ${X}\triangleq \{x_1,\dots,x_{k}\}$. Let $N\triangleq\{1,\dots,n\}$ and $K\triangleq\{1,\dots,k\}$. Initially, each user $i\in N$ has a subset ${X}_i$ of the packets in ${X}$, and ultimately, the user $i$ wants the rest of the packets $\overline{{X}}_i\triangleq {X}\setminus {X}_i$. The index set of packets in $X_i$ for each user $i$ is known by all other users. Also, without loss of generality, we assume that ${X}=\cup_{i\in N} {X}_i$. The objective of all users is to \emph{achieve omniscience}, i.e., to learn all packets in ${X}$, via exchanging their packets by broadcasting (coded or uncoded) packets. 


A subset $S$ of users in $N$ is a \emph{coalition} if $\cup_{i\in S} X_i=X$. We refer to any coalition $S\subset N$ as a \emph{minor coalition}, and refer to the coalition $N$ as the \emph{grand coalition}. Whenever we use the notation $S$ for a subset of users, we assume that $S$ is a coalition, unless explicitly noted otherwise. 



Let $\mathbb{Z}_{+}$ be the set of non-negative integers. For any $S\subseteq N$, a rate vector $r\triangleq [r_1,\dots,r_n]\in \mathbb{Z}_{+}^{n}$ is \emph{$S$-omniscience-achieving} if there exists a transmission scheme with each user $i\in S$ transmitting $r_{i}$ (coded or uncoded) packets such that all users in $S$ achieve omniscience, regardless of transmissions of the rest of the users. Note that, for any $S$-omniscience-achieving rate vector, random linear network coding (over a sufficiently large finite field) suffices as a transmission scheme for all users in $S$ to achieve omniscience (with any arbitrarily high probability)~\cite{CW:2014}. 

For any $S\subseteq N$, we denote by $\mathcal{R}_{S}$ the set of all $S$-omniscience-achieving rate vectors $r$ such that $r_i=0$ for all $i\not\in S$. For any arbitrary subset $S\subseteq N$ and any rate vector $r$, we define the sum-rate $r_S\triangleq \sum_{i\in S} r_i$ and $r_{\emptyset}\triangleq 0$. By a standard network coding argument~\cite{CW:2014}, for any $S\subseteq N$, $r\in\mathcal{R}_S$ iff $r_{\tilde{S}}\geq |\cap_{j\in S\setminus \tilde{S}} \overline{{X}}_j|$, for every (non-empty) $\tilde{S}\subset S$. 


We consider CDE under a monetary mechanism where there can be a payment from any user to any other user. For all $i,j\in N$, let ${p_{i,j}\geq 0}$ be the total payment from the user $i$ to the user $j$, and let $p_{i,i} = 0$. For a payment matrix $p\triangleq [p_{i,j}]$, let $p^{+}_i\triangleq \sum_{j\in N\setminus \{i\}} p_{j,i}$ and $p^{-}_i\triangleq \sum_{j\in N\setminus \{i\}} p_{i,j}$ be the total incoming payment of the user $i$ and the total outgoing payment of the user $i$, respectively. 

For any $S\subseteq N$, we denote by $\mathcal{P}_S$ the set of all payment matrices $p$ such that $p_{i,j}=0$ and $p_{j,i}=0$ for all $i\in S$, $j\not\in S$, i.e., there is no incoming payment to any user in $S$ from any user out of $S$ and there is no outgoing payment from any user in $S$ to any user out of $S$. For any $S\subseteq N$, we define the sum-payment $p_S\triangleq \sum_{i,j\in S} p_{i,j}$. Note that $\sum_{i\in S} p^{+}_i = \sum_{i\in S} p^{-}_i=p_S$ for all $p\in \mathcal{P}_S$.   
 
 

\begin{definition}[Utility]
For any $S\subseteq N$, any $r\in \mathcal{R}_S$, and any $p\in\mathcal{P}_S$, the utility of each user $i\in S$ is given by \[u_i(r,p)\triangleq (p^{+}_i-r_i)+(|\overline{X}_i|-p^{-}_i),\] where $u_i^{+}(r,p)\triangleq p^{+}_i-r_i$ is the net utility due to the user $i$'s contribution to the system, and $u_i^{-}(r,p)\triangleq |\overline{X}_i|-p^{-}_i$ is the net utility due to the system's contribution to the user $i$.
\end{definition}

Note that the cost per transmission and the value per packet are assumed to be unity for all users.

The two functions $u_i^{+}(r,p)$ and $u_i^{-}(r,p)$ motivate the notion of rationality defined as follows. 

\begin{definition}[Rationality] For any $S\subseteq N$, any $r\in \mathcal{R}_S$ and any $p\in \mathcal{P}_S$, the rate-payment pair $(r,p)$ is \emph{rational} if $u_i^{+}(r,p)\geq 0$ and $u_i^{-}(r,p)\geq 0$ for all $i\in S$.
\end{definition}


Hereafter, we focus on the rational rate-payment pairs only, and omit the term ``rational'' for brevity. 

We assume that all the users are \emph{selfish}, i.e., each user may or may not agree with its rate specified by a rate vector or its payments specified by a payment matrix. The goal is to find a rate-payment pair $(r,p)$, $r\in\mathcal{R}_{N}$ and $p\in \mathcal{P}_N$, under which $N$ is \emph{stable}. We formally define the notion of stability based on the utility function as follows. 



\begin{definition}[Stability]
For any rate-payment pair $(r,p)$, $r\in\mathcal{R}_N$ and $p\in\mathcal{P}_N$, $N$ is \emph{$(r,p)$-stable} if there is not a rate-payment pair $(\tilde{r},\tilde{p})$, $\tilde{r}\in\mathcal{R}_S$, and $\tilde{p}\in\mathcal{P}_S$, for some $S\subset N$, such that 
\begin{itemize}
\item $u_i(r,p)\leq u_i(\tilde{r},\tilde{p})$ for all $i\in S$, and
\item $u_i(r,p)< u_i(\tilde{r},\tilde{p})$ for some $i\in S$.
\end{itemize} 
\end{definition}

The $(r,p)$-stability of the grand coalition is equivalent to the Pareto optimality of $(r,p)$ over all minor coalitions.

\begin{definition}[Feasibility]
A rate-payment pair $(r,p)$ is \emph{feasible} if $N$ is $(r,p)$-stable.
\end{definition}


Note that a feasible solution guarantees that no minor coalition of users has incentive to break the grand coalition.



\begin{definition}[Optimality]
A feasible $(r,p)$ is \emph{optimal} if there is not a feasible $(\tilde{r},\tilde{p})$ such that $r_N>\tilde{r}_N$ or $p_N>\tilde{p}_N$.
\end{definition}

Note that, for an optimal solution, the sum-rate and the sum-payment are minimum among all feasible solutions.

The problem is to determine if an optimal solution exists for any given instance, and if so, to find such a solution.

{\SetAlgoNoLine
\begin{algorithm}[t!]
\caption{Algo1($n$, $k$, $\{\mathrm{U}_i\}_{i=1}^{n}$, $\mathbb{F}_q$)}
 $N\leftarrow \{1,\dots,n\}$, $K\leftarrow\{1,\dots,k\}$\\
 ${r}_i\leftarrow 0$ $\forall i\in N$, $p_{i,j}\leftarrow 0$ $\forall i,j\in N$\\
 $l\leftarrow 1$, $\mathrm{V}_0\leftarrow \emptyset$\\
 \While{$\dim(\mathrm{U}_i\cup \mathrm{V}_{l-1})< k$ \emph{for some} $i\in N$}{
  $T_l\leftarrow \{i\in N: \dim(\mathrm{U}_i\cup \mathrm{V}_{l-1})=\max_{i\in N} \dim(\mathrm{U}_i\cup \mathrm{V}_{l-1})\}$\\
  Select an arbitrary user $t\in T_l$\\ 
   $R_l\leftarrow \{i\in N: \mathrm{U}_{t}\not\subseteq \mathrm{span}(\mathrm{U}_i\cup \mathrm{V}_{l-1})\}$\\
   Select an encoding vector $v_{l}\in\mathbb{F}_q^{k}$ such that $v^{i}_{l}=0$ $\forall \{i\in K: u_i\not\in \mathrm{U}_{t}\}$ and $v_{l}\not\in \mathrm{span} (\cup_{i\in R_l}\mathrm{U}_i\cup\mathrm{V}_{l-1})$\\
   Have the user $t$ transmit the packet $y_{l}=\sum_{i\in K} v_{l}^{i} x_i$\\
   ${r}_{t}\leftarrow {r}_{t}+1$\\
   $p_{i,t}\leftarrow p_{i,t}+1/|R_l|$ $\forall i\in R_l$\\
   $\mathrm{V}_{l}\leftarrow \mathrm{V}_{l-1}\cup v_{l}$\\
$l\leftarrow l+1$
 }
 \Return ${r} = [{r}_i]_{i\in N}$ and $p = [p_{i,j}]_{i,j\in N}$
\end{algorithm}
}

\section{Proposed Algorithms}
\subsection{Algorithm~1}
In this section, we present an algorithm that, for any given instance, finds an optimal solution. 

The algorithm begins with an all-zero rate vector $r = [r_i]_{i\in N}$ and an all-zero payment matrix $p = [p_{i,j}]_{i,j\in N}$, operates in rounds, and updates $r$ and $p$ over the rounds. 

For any (uncoded) packet $x_i$, $i\in K$, denote the (unit) encoding vector of $x_i$ by $u_{i}\triangleq [u^{1}_{i},\dots,u^{k}_{i}]$, where $u^{i}_{i}=1$ and $u^{j}_{i}=0$ for all $j\neq i$. For any (linearly coded) packet $y_j \triangleq \sum_{i\in K} v^{i}_{j} x_i$, where $v^{i}_{j}\in\mathbb{F}_q$ (for some finite field $\mathbb{F}_q$), denote the encoding vector of $y_j$ by $v_{j} \triangleq [v_{j}^{1},\dots,v_{j}^{k}]$.

Let $\mathrm{U}_i$ be the set of (unit) encoding vectors of packets in $X_i$, and $\mathrm{V}_l$ be the set of encoding vectors of all packets being transmitted by the end of the round $l$. Let $\mathrm{V}_{0}\triangleq\emptyset$. We refer to $\mathrm{span}(\mathrm{U}_i\cup \mathrm{V}_l)$ and $\dim(\mathrm{U}_i\cup \mathrm{V}_l)$ as the \emph{knowledge} and the \emph{size of knowledge} of the user $i$ at the end of the round $l$, respectively, where $\mathrm{span}(\mathrm{V})$ and $\dim(\mathrm{V})$ denote the vector space of (linear) span (over $\mathbb{F}_q$) of a collection $\mathrm{V}$ of vectors in $\mathbb{F}_q^{k}$ and the dimension of $\mathrm{span}(\mathrm{V})$, respectively.  
 
Consider an arbitrary round $l>0$. Let $T_l$ be the set of all users $i$ with maximum $\dim(\mathrm{U}_i \cup \mathrm{V}_{l-1})$. In the round $l$, the algorithm first selects an arbitrary user $t\in T_l$, and then the user $t$ constructs (using its uncoded packets) and broadcasts a (coded) packet $y_{l}$ (with encoding vector $v_l$). 

Let $R_l$ be the set of all users $i$ such that $\mathrm{U}_{t}\not\subseteq\mathrm{span}(\mathrm{U}_i\cup\mathrm{V}_{l-1})$. The encoding vector $v_{l}$ of the packet $y_{l}$ satisfies two conditions: (i) $v_{l}^{i}=0$ $\forall \{i\in K: u_{i}\not\in \mathrm{U}_{t}\}$, and (ii) $v_{l}\not\in \mathrm{span}(\cup_{i\in R_l}\mathrm{U}_i\cup \mathrm{V}_{l-1})$. (Such a vector $v_{l}\in\mathbb{F}_q^k$ always exists and it can be found in polynomial time using a randomized or a deterministic algorithm so long as $q\geq n \cdot k$ or $q\geq n$, respectively~\cite{SSBR:2010}.) Note that $R_l$ is the set of all users $i$ whose knowledge at the beginning of the round $l$ is not a superset of (initial) knowledge of the transmitting user $t$, and the encoding vector $v_l$ of the packet $y_l$ being transmitted by the user $t$ in the round $l$ is not known to any user $i\in R_l$ at the beginning of the round $l$. Thus, the transmission of the packet $y_l$ increases the size of knowledge of any user $i\in R_l$ by one, and it does not change that of any user $i\not\in R_l$.

Next, the algorithm increments ${r}_{t}$ by $1$ and increments $p_{i,t}$ by $1/|R_l|$ for all $i\in R_l$. At the end of the round $l$, the algorithm augments $\mathrm{V}_{l-1}$ by $v_{l}$, and constructs $\mathrm{V}_{l}$, i.e., $\mathrm{V}_l = \mathrm{V}_{l-1}\cup \{v_l\}$. The rounds continue until the size of knowledge all users is $k$. Once the algorithm terminates, it returns the rate vector $r$ and the payment matrix $p$.

\begin{theorem}\label{thm:Solution}
The output of Algorithm~1 is optimal. 
\end{theorem}


\subsection{Algorithm~2}
In this section, we present an algorithm that for any given instance provides an optimal solution with maximum sum-utility and maximum min-utility among all optimal solutions. 

Algorithm~2 is similar to Algorithm~1, and the only difference is in the set of users that make payments and the update rule of the payments in each round. We assume that there is a \emph{broker} that collects the payment $p^{-}_i$ by each user $i$, and returns the payment $p^{+}_i$ to each user $i$. The algorithm begins with all-zero payment vectors $p^{+}$ and $p^{-}$, and updates these vectors over the rounds as follows. Consider an arbitrary round $l>0$. Let $P_l$ be the set of users with maximum $|\overline{X}_i|-p^{-}_i$. Assuming that the user $t$ transmits in the round $l$, the algorithm increments $p^{+}_t$ by $1$ and increments $p^{-}_i$ by $1/|P_l|$ for all $i\in P_l$. 

{\SetAlgoNoLine
\begin{algorithm}[t]
\caption{Algo2($n$, $k$, $\{\mathrm{U}_i\}_{i=1}^{n}$, $\mathbb{F}_q$)}
 $N\leftarrow \{1,\dots,n\}$, $K\leftarrow\{1,\dots,k\}$\\
 ${r}_i\leftarrow 0$, $p^{+}_{i}\leftarrow 0$, $p^{-}_{i}\leftarrow 0$ $\forall i\in N$\\
 $l\leftarrow 1$, $\mathrm{V}_0\leftarrow \emptyset$\\
 \While{$\dim(\mathrm{U}_i\cup \mathrm{V}_{l-1})< k$ \emph{for some} $i\in N$}{
  $T_l\leftarrow \{i\in N: \dim(\mathrm{U}_i\cup \mathrm{V}_{l-1})=\max_{i\in N} \dim(\mathrm{U}_i\cup \mathrm{V}_{l-1})\}$\\
  $P_l\leftarrow \{i\in N: |\overline{X}_i|-p^{-}_i = \max_{i\in N} (|\overline{X}_i|-p^{-}_i)\}$\\
  Select an arbitrary user $t\in T_l$\\ 
   $R_l\leftarrow \{i\in N: \mathrm{U}_{t}\not\subseteq \mathrm{span}(\mathrm{U}_i\cup \mathrm{V}_{l-1})\}$\\
   Select an encoding vector $v_{l}\in\mathbb{F}_q^{k}$ such that $v^{i}_{l}=0$ $\forall \{i\in K: u_i\not\in \mathrm{U}_{t}\}$ and $v_{l}\not\in \mathrm{span} (\cup_{i\in R_l}\mathrm{U}_i\cup\mathrm{V}_{l-1})$\\
   Have the user $t$ transmit the packet $y_{l}=\sum_{i\in K} v_{l}^{i} x_i$\\
   ${r}_{t}\leftarrow {r}_{t}+1$\\
   $p^{+}_{t}\leftarrow p^{+}_{t}+1$\\
   $p^{-}_{i}\leftarrow p^{-}_{i}+1/|P_l|$ $\forall i\in P_l$\\
   $\mathrm{V}_{l}\leftarrow \mathrm{V}_{l-1}\cup v_{l}$\\
$l\leftarrow l+1$
 }
 \Return ${r} = [{r}_i]_{i\in N}$ and $p = [p^{+}_{i},p^{-}_i]_{i\in N}$
\end{algorithm}
}

\begin{theorem}\label{thm:Solution2}
The output of Algorithm~2 is optimal. Moreover, the output of Algorithm~2 has maximum sum-utility and maximum min-utility among all optimal solutions.
\end{theorem}

\section{Proofs of Theorems}
\subsection{Proof of Theorem~\ref{thm:Solution}}
In this section, we reserve the notations $r$ and $p$ for the outputs of Algorithm~1.  

\begin{lemma}
$(r,p)$ is rational (i.e., $p^{+}_i\geq r_i$ and $|\overline{X}_i|\geq p^{-}_i$ for all $i\in N$).	
\end{lemma}

\begin{proof}
By the procedure of Algorithm~1, $p^{+}_i=r_i$ since the user $i$ receives one unit of payment for each transmission it makes, and $|\overline{X}_i|\geq p^{-}_i$ since the user $i$ pays at most one unit for each transmission that increases its size of knowledge, and it does not pay for any other transmission. 	
\end{proof}

Let $N_s$ be the $s$th subset of users that achieve omniscience simultaneously, and let $l_s$ be the round at which the users in $N_s$ achieve omniscience. Note that the sets $N_s$ are disjoint. Denote by $N^{(s)}$ the set of all users in $N_1,\dots,N_{s}$. Let $m$ be such that $N^{(m)}= N$. By using similar ideas as in the proof of~\cite[Lemma~4]{HS3:2016}, the following result can be shown. 

\begin{lemma}\label{lem:key1}
For any $s\in [m]$ and any $S\subseteq N^{(s)}$ such that $S\cap N_s\neq \emptyset$, we have $l_s\leq \tilde{r}_S$ for all $\tilde{r}\in \mathcal{R}_S$.	
\end{lemma}


\begin{proof}
Fix an arbitrary $s\in [m]$. Fix an arbitrary ${S\subseteq N^{(s)}}$ such that ${S\cap N_s\neq \emptyset}$, and an arbitrary ${\tilde{r}\in\mathcal{R}_S}$. Let ${\{y_l\}_{1\leq l\leq l_s}}$ be the set of the algorithm's choice of packets being transmitted from the round $1$ to the round $l_s$, and let ${\{v_l\}_{1\leq l\leq l_s}}$ be the set of encoding vectors of these packets. 


For any $S\subseteq N$, we say that a set of packets is \emph{$S$-transmittable} if the encoding vector of each packet in the set lies in $\mathrm{span}(\mathrm{U}_i)$ for some $i\in S$. Let ${\ell\triangleq\min(\tilde{r}_S,l_s)}$. We prove by induction (on $l$) that, for every $1\leq l\leq \ell$, there exists an $S$-transmittable set of $\tilde{r}_S-l+1$ packets such that if they were transmitted after the transmission of all the packets in the set $\{y_1,\dots,y_{l-1}\}$, then $S$ achieves omniscience. 

For the base case of $l=1$, there exists an $S$-transmittable set of $\tilde{r}_S$ packets such that if they were transmitted, then $S$ achieves omniscience (since $\tilde{r}\in\mathcal{R}_S$). Next, consider an arbitrary round $l$, $1<l\leq \ell$. Fix the set of packets $Y=\{y_1,\dots,y_{l-1}\}$. By the induction hypothesis, there exists an $S$-transmittable set of $\tilde{r}_S-l+1$ packets such that if they were transmitted after the transmission of $Y$, then $S$ achieves omniscience. Let $\tilde{{Y}}\triangleq \{\tilde{y}_{l},\dots,\tilde{y}_{\tilde{r}_S}\}$ and $\tilde{V} \triangleq \{\tilde{v}_{l},\dots,\tilde{v}_{\tilde{r}_S}\}$ be such a set of packets and the set of their encoding vectors, respectively. Assume that the algorithm selects the user $t$, which may or may not be in $S$, to transmit in the round $l$. 



Since \[\dim(\mathrm{U}_i\cup\mathrm{V}_{l-1})\geq k-\tilde{r}_S+l-1\] for all $i\in S$ (noting that, after the transmission of $Y\cup\tilde{Y}$, $S$ achieves omniscience), and \[\dim(\mathrm{U}_t\cup\mathrm{V}_{l-1})\geq \dim(\mathrm{U}_i\cup\mathrm{V}_{l-1})\] for all $i\in N$ (noting that, in the round $l$, the size of the knowledge of the user $t$ is greater than or equal to that of any other user $i\in N$), then \[\dim(\mathrm{U}_t\cup\mathrm{V}_{l-1})\geq k-\tilde{r}_S+l-1.\] If \[\dim(\mathrm{U}_t\cup\mathrm{V}_{l-1})= k-\tilde{r}_S+l-1,\] then the user $t$ cannot transmit in the round $l$ since the user $t$ needs the set of all the packets in $\tilde{Y}$ so as to achieve omniscience. This is, however, a contradiction (by assumption). Thus, \[\dim(\mathrm{U}_t\cup\mathrm{V}_{l-1})>k-\tilde{r}_S+l-1,\] and consequently, $\tilde{{Y}}$ contains some packet $\tilde{y}$ such that its encoding vector $\tilde{v}\in\mathrm{span}(\mathrm{U}_t\cup\mathrm{V}_{l-1})$. Fix such a packet $\tilde{y}$ and its encoding vector $\tilde{v}$. Note that, after the transmission of $Y\cup \tilde{Y}\setminus \{\tilde{y}\}$, the user $t$ achieves omniscience (i.e., $\dim(\mathrm{U}_t\cup \mathrm{V}_{l-1}\cup \tilde{V}\setminus \{\tilde{v}\})=k$), and any user $i\in S$, $i\neq t$, needs no more than one packet so as to achieve omniscience (i.e., $\dim(\mathrm{U}_i\cup \mathrm{V}_{l-1}\cup \tilde{V}\setminus \{\tilde{v}\})\geq k-1$ for all $i\in S$, $i\neq t$). (The deletion of one packet decreases the size of knowledge of any user by at most one.)  


Consider an arbitrary $i\in S$, $i\neq t$. We consider two cases: (i) ${v_l\in \mathrm{span}(\mathrm{U}_i\cup \mathrm{V}_{l-1}\cup \tilde{V}\setminus \{\tilde{v}\})}$, and (ii) ${v_l\not\in \mathrm{span}(\mathrm{U}_i\cup \mathrm{V}_{l-1}\cup \tilde{V}\setminus \{\tilde{v}\})}$. In the case (i), since \[v_l\in \mathrm{span}(\mathrm{U}_i\cup \mathrm{V}_{l-1}\cup \tilde{V}\setminus \{\tilde{v}\})\] and \[v_l\in \mathrm{span}(\mathrm{U}_t\cup \mathrm{V}_{l-1}),\] then \[\mathrm{span}(\mathrm{U}_t\cup\mathrm{V}_{l-1})\subseteq \mathrm{span}(\mathrm{U}_i\cup \mathrm{V}_{l-1}\cup \tilde{V}\setminus \{\tilde{v}\}),\] or equivalently, \[\mathrm{span}(\mathrm{U}_t\cup\mathrm{V}_{l-1}\cup \tilde{V}\setminus \{\tilde{v}\})\subseteq \mathrm{span}(\mathrm{U}_i\cup\mathrm{V}_{l-1}\cup \tilde{V}\setminus \{\tilde{v}\}).\] Thus, \[\dim(\mathrm{U}_t\cup\mathrm{V}_{l-1}\cup \tilde{V}\setminus \{\tilde{v}\})\leq \dim(\mathrm{U}_i\cup\mathrm{V}_{l-1}\cup \tilde{V}\setminus \{\tilde{v}\}),\] or equivalently, \[\dim(\mathrm{U}_i\cup\mathrm{V}_{l-1}\cup \tilde{V}\setminus \{\tilde{v}\}) = k\] since \[\dim(\mathrm{U}_t\cup\mathrm{V}_{l-1}\cup \tilde{V}\setminus \{\tilde{v}\})=k.\] Thus, after the transmission of $Y\cup \tilde{Y}\setminus \{\tilde{y}\}$, the user $i$ achieves omniscience. In the case (ii), since \[v_l\not\in \mathrm{span}(\mathrm{U}_i\cup \mathrm{V}_{l-1}\cup \tilde{V}\setminus \{\tilde{v}\})\] and \[\dim(\mathrm{U}_i\cup\mathrm{V}_{l-1}\cup \tilde{V}\setminus \{\tilde{v}\}) \geq k-1,\] then \[\dim(\mathrm{U}_i\cup\mathrm{V}_{l-1}\cup \{v_l\}\cup \tilde{V}\setminus \{\tilde{v}\}) = k.\] Thus, after the transmission of $Y\cup \{y_l\}\cup\tilde{Y}\setminus \{\tilde{y}\}$, the user $i$ achieves omniscience. By (i) and (ii), it follows that $S$ achieves omniscience after the transmission of $Y\cup\{y_l\}\cup\tilde{{Y}}\setminus \tilde{y}$. Thus, there exists an $S$-transmittable set of $\tilde{r}_S-l$ packets $\tilde{Y}\setminus \tilde{y}$ such that if they were transmitted after the transmission of $Y\cup y_l$, then $S$ achieves omniscience. This completes the inductive proof. 

From the above result, it follows that $S$ achieves omniscience by the algorithm's choice of packets $\{y_l\}_{1\leq l\leq \ell}$ being transmitted from the round $1$ to the round $\ell$. Now there are two cases: (i) $l_s>\tilde{r}_S$, and (ii) $l_s\leq \tilde{r}_S$. In the case (i), $\ell=\tilde{r}_S$, and hence, all users in $S$ must achieve omniscience by the round $\ell$ ($=\tilde{r}_S$). This is, however, a contradiction since some user(s) in $S$, particularly any user in $S\cap N_s$, achieves omniscience in the round $l_s$ ($>\ell$) (by definition). Note that $S\cap N_s\neq \emptyset$ (by assumption). In the case (ii), $\ell=l_s$, and the lemma follows directly. This completes the proof.
\end{proof}

\begin{lemma}\label{lem:feasibility1}
$(r,p)$ is feasible (i.e., $N$ is $(r,p)$-stable). 
\end{lemma}

\begin{proof}
The proof follows by contradiction. Suppose that $(r,p)$ is not feasible (i.e., $N$ is not $(r,p)$-stable). Thus, there exists $\tilde{r}\in \mathcal{R}_S$ and $\tilde{p}\in \mathcal{P}_S$ for some $S\subset N$ such that $u_i(r,p)\leq u_i(\tilde{r},\tilde{p})$ for all $i\in S$, and $u_i(r,p)<u_i(\tilde{r},\tilde{p})$ for some $i\in S$. Thus, 
\[\sum_{i\in S} u_i(\tilde{r},\tilde{p})>\sum_{i\in S} u_i(r,p).\] Note that \[\sum_{i\in S} u_i(\tilde{r},\tilde{p}) = \sum_{i\in S} p^{+}_i - \sum_{i\in S} \tilde{r}_i +\sum_{i\in S} |\overline{X}_i|-\sum_{i\in S} p^{-}_i.\] Since $\sum_{i\in S} p^{+}_i = \sum_{i\in S} p^{-}_i$ for all $p\in \mathcal{P}_S$, then \[\sum_{i\in S} u_i(\tilde{r},\tilde{p}) = \sum_{i\in S} |\overline{X}_i| - \sum_{i\in S} \tilde{r}_i.\]	Since $r_i = p^{+}_i$ for all $i\in N$, then $\sum_{i\in S} r_i = \sum_{i\in S} p^{+}_i$. Thus, \[\sum_{i\in S} u_i(r,p) = \sum_{i\in S} |\overline{X}_i| - \sum_{i\in S} p^{-}_i.\] Putting these arguments together, we get 
\begin{equation}\label{eq:Eq4}
\sum_{i\in S} p^{-}_i>\sum_{i\in S} \tilde{r}_i.
\end{equation} Let $s\in [m]$ be such that $S\subseteq N^{(s)}$ and $S\cap N_s\neq \emptyset$. Note that all the users in $S$ achieve omniscience by the round $l_s$. By the structure of the proposed algorithm, one unit of payment is made in each round (each user in $R_l$ pays $1/|R_l|$ units of payment in the round $l$), and no user pays in any round after it achieves omniscience (if the user $i$ is complete at the beginning of the round $l$, then $i\not\in R_l$). Thus, it is easy to see that \[\sum_{i\in S} p^{-}_i\leq l_s.\] Moreover, by the result of Lemma~\ref{lem:key1}, it follows that \[
l_s\leq \sum_{i\in S} \tilde{r}_i\] for all $\tilde{r}\in \mathcal{R}_S$. By combining these two inequalities, we get
\begin{equation}\label{eq:Eq7}
\sum_{i\in S} p^{-}_i \leq 	\sum_{i\in S} \tilde{r}_i.
\end{equation} By comparing~\eqref{eq:Eq4} and~\eqref{eq:Eq7}, we arrive at a contradiction. Thus, $N$ is $(r,p)$-stable, as was to be shown.
\end{proof}

\begin{lemma}[\cite{SSBR:2010}]\label{lem:sumratemin}
For any $\tilde{r}\in \mathcal{R}_N$, we have $\tilde{r}_N\geq r_N$.	
\end{lemma}

\begin{proof}
The proof can be found in \cite{SSBR:2010}.	
\end{proof}

\begin{lemma}\label{lem:optimality1}
$(r,p)$ is optimal (i.e., there is not a feasible $(\tilde{r},\tilde{p})$ such that $r_N>\tilde{r}_N$ or $p_N>\tilde{p}_N$).	
\end{lemma}

\begin{proof}
Consider an arbitrary feasible $(\tilde{r},\tilde{p})$, $\tilde{r}\in \mathcal{R}_N$ and $\tilde{p}\in \mathcal{P}_N$. We shall show that $\tilde{r}_N\geq r_N$ and $\tilde{p}_N\geq p_N$. By Lemma~\ref{lem:sumratemin}, $\tilde{r}_N\geq r_N$ for all $\tilde{r}\in \mathcal{R}_N$. Since $(\tilde{r},\tilde{p})$ is feasible, then $(\tilde{r},\tilde{p})$ is rational. Thus, $\tilde{p}^{+}_i\geq \tilde{r}_i$ for all $i\in N$, and consequently, $\tilde{p}_N\geq \tilde{r}_N$. Note that $p_N = r_N$ since $p^{+}_i = r_i$. Thus, $\tilde{p}_N\geq \tilde{r}_N\geq r_N = p_N$. This completes the proof.   	
\end{proof}

\subsection{Proof of Theorem~\ref{thm:Solution2}}
In this section, we reserve the notations $r$ and $p$ for the outputs of Algorithm~2.

\begin{lemma}\label{lem:rationalityAlgo2}
$(r,p)$ is rational.	
\end{lemma}

\begin{proof}
Let $r_i(l)$, $p^{+}_i(l)$, and $p^{-}_i(l)$ be $r_i$, $p^{+}_i$, and $p^{-}_i$ at the end of the round $l-1$, respectively. Note that $r_i = r_i(l_m+1)$, $p^{+}_i = p^{+}_i(l_m+1)$, and $p^{-}_i = p^{-}_i(l_m+1)$. We will show that $p^{+}_i(l)\geq r_i(l)$ and $|\overline{X}_i|\geq p^{-}_i(l)$ for all $i\in N$ and all $l\in [l_m+1]$. Fix an arbitrary $l\in [l_m+1]$. By the procedure of Algorithm~2, $p^{+}_i(l) = r_i(l)$, and particularly, $p^{+}_i = r_i$. We next show that $|\overline{X}_i|\geq p^{-}_i(l)$. The proof follows by contradiction. Suppose that $|\overline{X}_i|< p^{-}_i(l)$ for some $i$. Note that \[\max_{i\in N} |\overline{X}_i| - p^{-}_i(l) = k - \min_{i\in N} (|X_i|+p^{-}_i(l)).\] Thus, \[P_l = \{i\in N: |X_i|+p^{-}_i(l) = \min_{i\in N} (|X_i|+p^{-}_i(l))\}.\] By the procedure of Algorithm~2, $|X_i|+p^{-}_i(l)$ are the same for all $i$ such that $p^{-}_i(l)>0$, and $|X_i|+p^{-}_i(l) = |X_i|\leq k$ for all $i$ such that $p^{-}_i(l) = 0$. Since $|X_i|+p^{-}_i(l)>k$ for some $i$ (by assumption), then $|X_i|+p^{-}_i(l)>k$ for all $i$, and consequently, $p^{-}_i(l)>0$ for all $i$ (since $|X_i|\leq k$ for all $i$). Since $p^{-}_i(l)$ is non-decreasing in $l$ for all $i$, then $|X_i|+p^{-}_i>k$ for all $i$, or equivalently, $p^{-}_i>|\overline{X}_i|$ for all $i$. Thus, \[\sum_{i\in N} p^{-}_i>\sum_{i\in N} |\overline{X}_i|,\] and consequently, \[r_N>\sum_{i\in N} |\overline{X}_i|\] since \[\sum_{i\in N} p^{-}_i = \sum_{i\in N} p^{+}_i =  r_N.\] This is, however, a contradiction since \[r_N\leq \min_{i\in N} |\overline{X}_i|+\max_{i\in N} |\overline{X}_i|\] (by the result of \cite[Lemma~3]{RSS:2010}), and consequently, \[r_N\leq \sum_{i\in N} |\overline{X}_i|.\] Thus, $|\overline{X}_i|\geq p^{-}_i(l)$ for all $i$ and all $l$, and particularly, $|\overline{X}_i|\geq p^{-}_i$ for all $i$. This completes the proof.
\end{proof}



\begin{lemma}
$(r,p)$ is feasible. 
\end{lemma}

\begin{proof}
Take an arbitrary $S$ such that $\mathcal{R}_S\neq \emptyset$ (i.e., all users in $S$ can achieve omniscience together). By the same argument as in the proof of Lemma~\ref{lem:feasibility1}, it suffices to show that \[\sum_{i\in S} p^{-}_i\leq \tilde{r}_S\] for all $\tilde{r}\in\mathcal{R}_S$. Run Algorithm~2 over the set $S$, and denote by $(\tilde{r},\tilde{p})$ the output. Let $\tilde{Y}=\{\tilde{y}_1,\dots,\tilde{y}_{\tilde{r}_S}\}$ and $\tilde{V}=\{\tilde{v}_1,\dots,\tilde{v}_{\tilde{r}_S}\}$ be the set of all packets being transmitted from the round $1$ to the round $\tilde{r}_S$ and their encoding vectors, respectively. Note that $\tilde{r}_S$ is the minimum sum-rate that all users in $S$ can achieve omniscience (by Lemma~\ref{lem:sumratemin}). Assume, without loss of generality, that $|X_1|\leq |X_2|\leq \dots\leq |X_n|$. Define $i^{\star} \triangleq \min_{i\in S} i$, and $S^{\star}\triangleq \{i^{\star},\dots,n\}$. Since $S\subseteq S^{\star}$, then $\mathcal{R}_{S^{\star}}\neq \emptyset$ (i.e., all users in $S^{\star}$ can achieve omniscience together). Moreover, run Algorithm~2 over the set $S^{\star}$, and denote by $(r^{\star},p^{\star})$ the output. Note that $p^{\star}_{S^{\star}} = r^{\star}_{S^{\star}}$ (by the result of Lemma~\ref{lem:rationalityAlgo2}). Let $Y^{\star}=\{y^{\star}_1,\dots,y^{\star}_{r^{\star}_{S^{\star}}}\}$ and $V^{\star}=\{v^{\star}_1,\dots,v^{\star}_{r^{\star}_{S^{\star}}}\}$ be the set of all packets being transmitted from the round $1$ to the round $r^{\star}_{S^{\star}}$ and their encoding vectors, respectively.  

First, we show that \[r^{\star}_{S^{\star}}\leq \tilde{r}_S.\] To do so, it suffices to show that all users in $S^{\star}\setminus S$ achieve omniscience after the reception of all packets in $\tilde{Y}$. The proof follows by contradiction. Consider an arbitrary user $i\in S^{\star}\setminus S$. Suppose that the user $i$ does not achieve omniscience after the reception of all packets in $\tilde{Y}$, i.e., $\dim(\mathrm{U}_i \cup \tilde{V})<k$. Since $\dim(\mathrm{U}_i) \geq \dim(\mathrm{U}_{i^{\star}})$ and $\dim(\mathrm{U}_{i^{\star}})\geq k-\tilde{r}_S$, then $\dim(\mathrm{U}_i)\geq k-\tilde{r}_S$. Thus, there exists some round $l$ such that the encoding vector $\tilde{v}_l$ of the packet $\tilde{y}_l$ being transmitted by some user $t\in S$ is in the knowledge set of the user $i$ prior to the round $l$, i.e., \[\mathrm{span}(\mathrm{U}_{t}) \subseteq \mathrm{span}(\mathrm{U}_i\cup \{\tilde{v}_{1},\dots,\tilde{v}_{l-1}\}),\] and consequently, \[\mathrm{span}(\mathrm{U}_{t}\cup \{\tilde{v}_{1},\dots,\tilde{v}_{l-1}\}) \subseteq \mathrm{span}(\mathrm{U}_i\cup \{\tilde{v}_{1},\dots,\tilde{v}_{l-1}\}).\] Thus, \[\mathrm{span}(\mathrm{U}_{t}\cup \tilde{V}) \subseteq \mathrm{span}(\mathrm{U}_i\cup \tilde{V}).\] Since $\dim(\mathrm{U}_{t}\cup \tilde{V})=k$ and $\dim(\mathrm{U}_{i}\cup \tilde{V})\geq \dim(\mathrm{U}_{t}\cup \tilde{V})$, then $\dim(\mathrm{U}_{i}\cup \tilde{V})=k$. This is, however, a contradiction since $\dim(\mathrm{U}_{i}\cup \tilde{V})<k$ (by assumption). Thus, all users in $S^{\star}\setminus S$ achieve omniscience after the reception of all packets in $\tilde{Y}$, and so, $r^{\star}_{S^{\star}}\leq \tilde{r}_S$.

Next, we show that \[\sum_{i\in S^{\star}} p^{-}_i \leq p^{\star}_{S^{\star}}.\] If $S^{\star} = N$, then \[\sum_{i\in S^{\star}} p^{-}_i  = p_N = r_N = r^{\star}_N= p^{\star}_N = p^{\star}_{S^{\star}}.\] Now assume that $S^{\star}\neq N$. If for some $l$, the packet $y^{\star}_l$ being transmitted by the user $t\in S^{\star}$ does not increase the size of knowledge of the user $i\in N\setminus S^{\star}$ such that $\dim(\mathrm{U}_{i}\cup V^{\star})<k$, then \[\mathrm{span}(\mathrm{U}_t)\subseteq \mathrm{span}(\mathrm{U}_i\cup \{v^{\star}_1,\dots,v^{\star}_{l-1}\}),\] and consequently, \[\mathrm{span}(\mathrm{U}_t\cup \{v^{\star}_1,\dots,v^{\star}_{l-1}\})\subseteq \mathrm{span}(\mathrm{U}_i\cup \{v^{\star}_1,\dots,v^{\star}_{l-1}\}).\] Thus, \[\mathrm{span}(\mathrm{U}_t\cup V^{\star})\subseteq \mathrm{span}(\mathrm{U}_i\cup V^{\star}).\] Since $\dim(\mathrm{U}_{t}\cup V^{\star})=k$ and $\dim(\mathrm{U}_{i}\cup V^{\star})\geq \dim(\mathrm{U}_{t}\cup V^{\star})$, then $\dim(\mathrm{U}_{i}\cup V^{\star})=k$. This yields a contradiction since $\dim(\mathrm{U}_{i}\cup V^{\star})<k$ (by assumption). Thus, the packet $y^{\star}_l$ (for any $l$) increases the size of knowledge of all users in $N\setminus S^{\star}$ that do not achieve omniscience after the reception of all packets $y^{\star}_1,\dots,y^{\star}_{r^{\star}_{S^{\star}}}$. 

Since the size of knowledge of each user $i\in N\setminus S^{\star}$ after the reception of all packets in $Y^{\star}$ is $\min\{|X_i|+r^{\star}_{S^{\star}},k\}$, then the user $i$ needs $k-\min\{|X_i|+r^{\star}_{S^{\star}},k\}$ ($\leq k-\min\{|X_1|+r^{\star}_{S^{\star}},k\}$) more packets to achieve omniscience. Thus, if the users in $S^{\star}$ continue to make transmissions after they all achieve omniscience, all users in $N\setminus S^{\star}$ achieve omniscience after the reception of at most $k-\min\{|X_1|+r^{\star}_{S^{\star}},k\}$ more packets. Thus, all users in $N$ achieve omniscience with at most $r^{\star}_{S^{\star}}+k-\min\{|X_1|+r^{\star}_{S^{\star}},k\}$ total transmissions. Since $r_N$ is the minimum sum-rate for all users in $N$ to achieve omniscience, then \[r^{\star}_{S^{\star}}+k-\min\{|X_1|+r^{\star}_{S^{\star}},k\}\geq r_N.\] We consider two cases: (i) $|X_1|+r^{\star}_{S^{\star}}\geq k$, and (ii) $|X_1|+r^{\star}_{S^{\star}}< k$. 

In the case~(i), we have \[r^{\star}_{S^{\star}}\geq r_N=p_N\geq \sum_{i\in S^{\star}} p^{-}_i.\] In the case~(ii), we have \[r^{\star}_{S^{\star}}+k-|X_1|-r^{\star}_{S^{\star}} = k-|X_1|\geq r_N.\] Since $r_N\geq k-|X_1|$ (otherwise, the user $1$ cannot achieve omniscience), then $r_N = k-|X_1|$. Let $c \triangleq \min_{i\in N} (|X_i|+p^{-}_i)$. If $c<|X_{i^{\star}}|$, then \[\sum_{i\in S^{\star}} p^{-}_i\leq p^{\star}_{S^{\star}}\] since $\sum_{i\in S^{\star}} p^{-}_i=0$. Now, assume that $c\geq |X_{i^{\star}}|$. Recall that $|X_1|\leq |X_2|\leq \dots\leq |X_{i^{\star}}|\leq \dots\leq |X_n|$ (by assumption). Thus, $c\geq |X_i|$ for all $i\in N\setminus S^{\star}$. Note that \[\sum_{i\in S^{\star}} p^{-}_i = r_N-\sum_{i\in N\setminus S^{\star}} (c-|X_i|)\] and \[p^{\star}_{S^{\star}}=r^{\star}_{S^{\star}}.\] We need to show that \[\sum_{i\in S^{\star}} p^{-}_i\leq p^{\star}_{S^{\star}}.\] Thus it suffices to show that \[r_N-\sum_{i\in N\setminus S^{\star}} (c-|X_i|)\leq r^{\star}_{S^{\star}}.\] The proof follows by contradiction. Suppose that \[r_N-\sum_{i\in N\setminus S^{\star}} (c-|X_i|)> r^{\star}_{S^{\star}}.\] Since $r_N = k-|X_1|$ and $r^{\star}_{S^{\star}}\geq k-|X_{i^{\star}}|$ (otherwise, the user $i^{\star}$ cannot achieve omniscience), then \[k-|X_1|-\sum_{i\in N\setminus S^{\star}} (c-|X_i|)> r^{\star}_{S^{\star}}\geq k-|X_{i^{\star}}|,\] and consequently, \[|X_{i^{\star}}|>|X_1|+\sum_{i\in N\setminus S^{\star}} (c-|X_i|).\] Since \[\sum_{i\in N\setminus S^{\star}} (c-|X_i|) = (i^{\star}-1)c-(|X_1|+\dots+|X_{i^{\star}-1}|),\] then \[|X_2|+\dots+|X_{i^{\star}}|>(i^{\star}-1)c.\] This is, however, a contradiction since $c\geq |X_i|$ for all $i\in [i^{\star}]$ (by assumption), and so, \[(i^{\star}-1)c\geq |X_2|+\dots+|X_{i^{\star}}|.\] Thus, \[\sum_{i\in S^{\star}} p^{-}_i \leq p^{\star}_{S^{\star}}.\] Moreover, \[\sum_{i\in S} p^{-}_i  \leq \sum_{i\in S^{\star}} p^{-}_i\] since $S\subseteq S^{\star}$ (by definition). By combining the above arguments, it then follows that \[\sum_{i\in S} p^{-}_i \leq \tilde{r}_S,\] as was to be shown.
\end{proof}

\begin{lemma}\label{lem:feasibilityoptimality2}
$(r,p)$ is optimal. 
\end{lemma}
	
\begin{proof}
The proof follows from the same argument as in the proof of Lemma~\ref{lem:optimality1}, and hence omitted to avoid repetition. 	
\end{proof}

\begin{lemma}
For any optimal $(\tilde{r},\tilde{p})$, we have $\sum_{i\in N} u_i(r,p)= \sum_{i\in N} u_i(\tilde{r},\tilde{p})$ and $\min_{i\in N} u_i(r,p)\geq \min_{i\in N} u_i(\tilde{r},\tilde{p})$.	
\end{lemma}
	
\begin{proof}
The proof of the first part (i.e., maximum sum-utility) is straightforward. Take an arbitrary optimal $(\tilde{r},\tilde{p})$. Since $\tilde{r}_N = r_N$ and $p,\tilde{p}\in \mathcal{P}_N$, then \[\sum_{i\in N} u_i(\tilde{r},\tilde{p})=\sum_{i\in N} |\overline{X}_i|-\tilde{r}_N= \sum_{i\in N} |\overline{X}_i|-r_N = \sum_{i\in N} u_i(r,p).\] For the proof of the second part (i.e., maximum min-utility), we need to show that \[\min_{i\in N} u_i(\tilde{r},\tilde{p})\leq \min_{i\in N} u_i(r,p)\] for any optimal $(\tilde{r},\tilde{p})$. Take an arbitrary optimal $(\tilde{r},\tilde{p})$. Since $\tilde{p}^{+}_i = \tilde{r}_i$ (otherwise, $\tilde{p}_N>r_N=p_N$ since $\tilde{p}^{+}_i\geq \tilde{r}_i$ (by rationality of $(\tilde{r},\tilde{p})$), and so, $(\tilde{r},\tilde{p})$ cannot be optimal), then $u_i(\tilde{r},\tilde{p}) = |\overline{X}_i|-\tilde{p}^{-}_i$. Note that $\tilde{p}_N=p_N$. Let $c \triangleq \min_{i\in N} (|X_i|+p^{-}_i)$. Note that $|\overline{X}_i|-p^{-}_i=k-c$ if $c\geq |X_i|$, and $|\overline{X}_i|-p^{-}_i=|\overline{X}_i|=k-|X_i|$ if $c< |X_i|$. Thus, $u_i(r,p)=k-\max\{c,|X_i|\}$ for all $i\in N$. Since $|X_n|\geq |X_i|$ for all $i\in N$ (by assumption), then it follows that \[\min_{i\in N} u_i(r,p)=k-\max\{c,|X_n|\}.\] We consider two cases: (i) $c<|X_n|$, and (ii) $c\geq |X_n|$. 

In the case (i), $\min_{i\in N} u_i(r,p)=k-|X_n|=|\overline{X}_n|$. If $\tilde{p}^{-}_n>0$, then \[u_n(\tilde{r},\tilde{p})=|\overline{X}_n|-\tilde{p}^{-}_n<|\overline{X}_n|=\min_{i\in N} u_i(r,p).\] If $\tilde{p}^{-}_n=0$, then $u_n(\tilde{r},\tilde{p})=|\overline{X}_n|$, and consequently, \[\min_{i\in N} u_i(\tilde{r},\tilde{p})\leq u_n(\tilde{r},\tilde{p})=\min_{i\in N} u_i(r,p)=k-c.\] In the case (ii), $\min_{i\in N} u_i(r,p)=k-c$. Suppose that $\min_{i\in N} u_i(\tilde{r},\tilde{p})>\min_{i\in N} u_i(r,p)$. Let $j\in N$ be such that $|\overline{X}_{j}|-\tilde{p}^{-}_{j}=\min_{i\in N} u_i(\tilde{r},\tilde{p})$. Thus, $|\overline{X}_{j}|-\tilde{p}^{-}_{j}>k-c$. Since $|\overline{X}_i|-\tilde{p}^{-}_i\geq |\overline{X}_{j}|-\tilde{p}^{-}_{j}$ for all $i\in N$, then $|\overline{X}_i|-\tilde{p}^{-}_i>k-c$. Thus, 
\begin{eqnarray*}
\sum_{i\in N} |\overline{X}_i|-\sum_{i\in N} \tilde{p}^{-}_i &=& \sum_{i\in N} |\overline{X}_i|-\tilde{p}_N\\ &=& \sum_{i\in N} |\overline{X}_i|-\tilde{r}_N\\ &=& \sum_{i\in N} |\overline{X}_i|-r_N\\ &=& nk-\sum_{i\in N} |X_i|-r_N\\ &>& nk-nc,	
\end{eqnarray*}
or equivalently, $(\sum_{i\in N} |X_i|+r_N)/n<c$. Since $c=\min_{i\in N} (|X_i|+p^{-}_i)$ (by definition) and $c\geq |X_n|$ (by assumption), then it is easy to see that $c=(\sum_{i\in N} |X_i|+r_N)/n$. This is a contradiction since $(\sum_{i\in N} |X_i|+r_N)/n<c$. Thus, $\min_{i\in N} u_i(\tilde{r},\tilde{p})\leq\min_{i\in N} u_i(r,p)$. This completes the proof.	
\end{proof}

\bibliographystyle{IEEEtran}
\bibliography{CDERefs}

\end{document}